\documentclass[11pt,onecolumn,twoside,aps,prl,nopacs,showkeys,superscriptaddress,a4paper]{revtex4-1}
\usepackage{graphicx}
\usepackage{amsmath,amssymb}
\usepackage[usenames,dvipsnames]{color}
\usepackage{bbm}
\usepackage{hyperref}
\usepackage{theorem}
\usepackage{latexsym}

\newtheorem{theorem}{Theorem}

\newlength{\blank}
\newenvironment{proof}[1][{\hspace{-\blank}}]{{\noindent\textbf{Proof~{#1}.\ }}}{\hfill\qed}
\newcommand{\ket}[1]{|#1\rangle}
\newcommand{\bra}[1]{\langle#1|}

\mathchardef\ordinarycolon\mathcode`\:
\mathcode`\:=\string"8000
\def\vcentcolon{\mathrel{\mathop\ordinarycolon}}
\begingroup \catcode`\:=\active
  \lowercase{\endgroup
  \let :\vcentcolon
  }

\newcommand{\nc}{\newcommand}
\nc{\rnc}{\renewcommand}
\nc{\beq}{\begin{equation}}
\nc{\eeq}{{\end{equation}}}
\nc{\beqa}{\begin{eqnarray}}
\nc{\eeqa}{\end{eqnarray}}
\nc{\lbar}[1]{\overline{#1}}
\nc{\ketbra}[2]{|#1\rangle\!\langle#2|}
\nc{\proj}[1]{| #1\rangle\!\langle #1 |}
\nc{\avg}[1]{\langle#1\rangle}
\nc{\Rank}{\operatorname{Rank}}
\nc{\smfrac}[2]{\mbox{$\frac{#1}{#2}$}}
\nc{\tr}{\operatorname{Tr}}
\nc{\ox}{\otimes}
\nc{\dg}{\dagger}
\nc{\dn}{\downarrow}
\nc{\cA}{\mathcal{A}}
\nc{\cB}{\mathcal{B}}
\nc{\cC}{\mathcal{C}}
\nc{\cD}{\mathcal{D}}
\nc{\cE}{\mathcal{E}}
\nc{\cF}{\mathcal{F}}
\nc{\cG}{\mathcal{G}}
\nc{\cH}{\mathcal{H}}
\nc{\cI}{\mathcal{I}}
\nc{\cJ}{\mathcal{J}}
\nc{\cK}{\mathcal{K}}
\nc{\cL}{\mathcal{L}}
\nc{\cM}{\mathcal{M}}
\nc{\cN}{\mathcal{N}}
\nc{\cO}{\mathcal{O}}
\nc{\cP}{\mathcal{P}}
\nc{\cR}{\mathcal{R}}
\nc{\cS}{\mathcal{S}}
\nc{\cT}{\mathcal{T}}
\nc{\cX}{\mathcal{X}}
\nc{\cZ}{\mathcal{Z}}
\nc{\csupp}{{\operatorname{csupp}}}
\nc{\qsupp}{{\operatorname{qsupp}}}
\nc{\var}{\operatorname{var}}
\nc{\rar}{\rightarrow}
\nc{\lrar}{\longrightarrow}
\nc{\polylog}{\operatorname{polylog}}
\nc{\1}{{\openone}}
\nc{\id}{{\operatorname{id}}}

\nc{\RR}{{{\mathbb R}}}
\nc{\CC}{{{\mathbb C}}}
\nc{\FF}{{{\mathbb F}}}
\nc{\NN}{{{\mathbb N}}}
\nc{\ZZ}{{{\mathbb Z}}}
\nc{\PP}{{{\mathbb P}}}
\nc{\QQ}{{{\mathbb Q}}}
\nc{\UU}{{{\mathbb U}}}
\nc{\EE}{{{\mathbb E}}}

\nc{\qed}{{$\hfill\Box$}}

\begin{document}

\title{Interferometric visibility and coherence}

\author{Tanmoy Biswas}
\email{tanmoy.biswas23@gmail.com}
\affiliation{Department of Physical Sciences, IISER Kolkata, Mohanpur 741246, West Bengal, India}

\author{Mar\'ia Garc\'ia D\'iaz}
\email{maria.garciadia@e-campus.uab.cat}
\affiliation{Departament de F\'{\i}sica: Grup d'Informaci\'{o} Qu\`{a}ntica, %
Universitat Aut\`{o}noma de Barcelona, ES-08193 Bellaterra (Barcelona), Spain}

\author{Andreas Winter}
\email{andreas.winter@uab.cat}
\affiliation{Departament de F\'{\i}sica: Grup d'Informaci\'{o} Qu\`{a}ntica, %
Universitat Aut\`{o}noma de Barcelona, ES-08193 Bellaterra (Barcelona), Spain}
\affiliation{ICREA---Instituci\'{o} Catalana de Recerca i Estudis Avan\c{c}ats, %
             Pg.~Lluis Companys, 23, ES-08001 Barcelona, Spain}

\date{20 June 2017}


\begin{abstract}
Recently, the basic concept of quantum coherence (or superposition)
has gained a lot of renewed attention, 
after Baumgratz \emph{et al.} [PRL 113:140401 (2014)],
following \AA{}berg [arXiv:quant-ph/0612146], have proposed a resource
theoretic approach to quantify it. This has resulted in a large number of
papers and preprints exploring various coherence monotones, and debating
possible forms for the resource theory. 
Here we take the view that the operational foundation of coherence in
a state, be it quantum or otherwise wave mechanical, lies in the observation 
of interference effects. 

Our approach here is to consider an idealised
multi-path interferometer, with a suitable detector, in such a way that the
visibility of the interference pattern provides a quantitative expression
of the amount of coherence in a given probe state.
We present a general framework of deriving coherence measures 
from visibility, and demonstrate it by analysing several concrete visibility
parameters, recovering some known coherence measures and obtaining some new ones.
\end{abstract}

\maketitle

\textbf{Introduction.}---%
The physics of constructive and destructive interference of waves, along with the
concept of coherence, has been well-understood since the 19th century. 
With the advent of quantum mechanics, these studies have assumed 
a fundamental quality as in quantum theory the superposition principle
applies to everything, and the presence of quantum
coherence is the basic hallmark of departure from classical physics.
Recently, Baumgratz \emph{et al.}~\cite{BCP14}, following \AA{}berg's
earlier work \cite{Aaberg2006}, have launched a flurry of new activity on 
coherence by attempting to cast it as a resource theory and introducing
a number of tasks and monotones \cite{coherence-review,MarvianSpekkens-review}.

It has, however, remained largely unclear what this resource of ``coherence'' 
is about and how it relates to theories of asymmetry, among others \cite{MarvianSpekkens-review}.
To make contact with the operational foundations of coherence, we go back
to its very definition, the observability of an interference pattern in
a suitable experiment.
Our present approach is to consider an idealised multi-path interferometer,
which receives the state $\rho$ under consideration at the input.
The experimenter is at liberty to put phase plates into each
of the paths, and to construct a detector (a general beam splitter with 
detection of the output beams). The interference pattern, i.e. the response
of the fixed detector as a function of the multiple phases, is the signature
of coherence: the more it fluctuates, intuitively, the more coherent
is the state. The degree of fluctuation, aka 
\emph{visibility}, quantifies the strength of interference.

The idea in the present paper is that coherence is the potential
of a state to yield visible fringes in a suitable experiment. 
Hence, we propose to optimise the visibility over all possible
detectors, to obtain a measure of coherence of the original state.
Indeed, we prove that under mild assumptions, every visibility parameter
yields a coherence measure in this way, strongly monotonic 
under a certain class of incoherence-preserving operations. 
We illustrate our theory with concrete examples of visibility parameters.


\bigskip
\textbf{Interferometers and visibility.}---%
Consider a multi-path interferometer,
in which a single particle can be in one of $d$ paths, denoting
the spatial variable by orthogonal vectors $\ket{j}$, $j=1,\ldots,d$,
spanning a $d$-dimensional Hilbert space $\cH$.
\begin{figure}[ht]
  \includegraphics[width=0.48\textwidth]{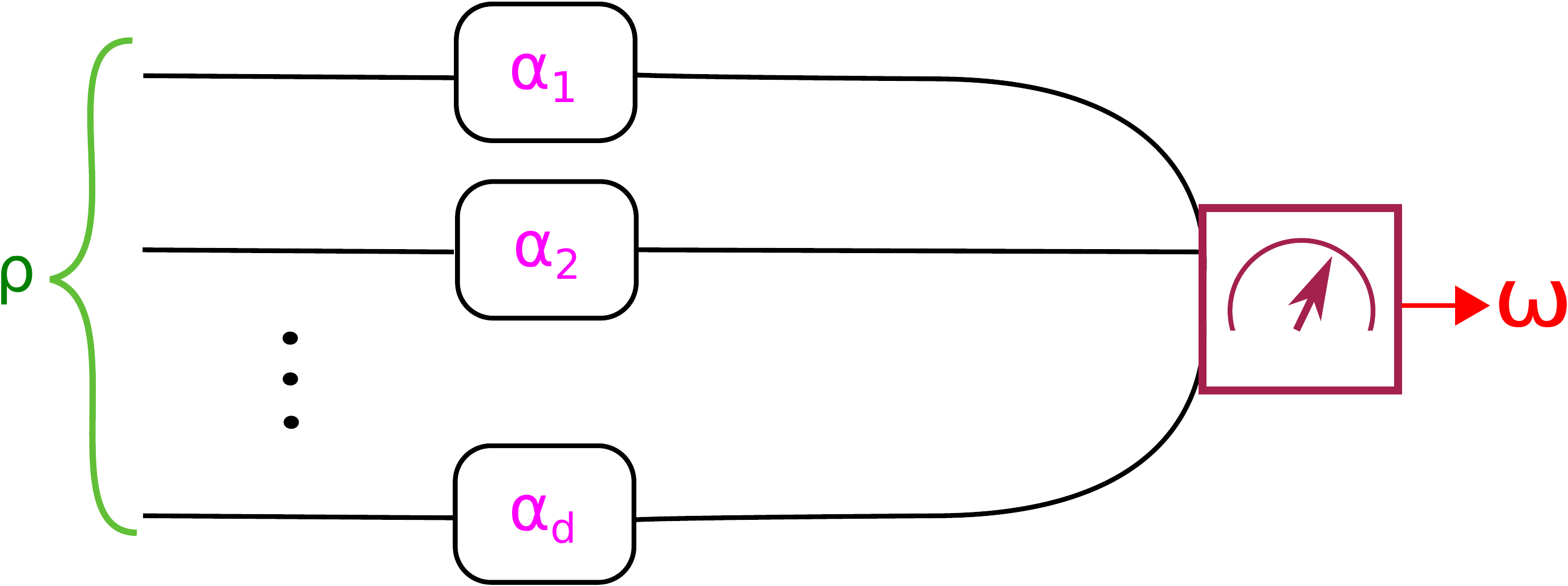}
  \caption{Schematic of a multi-path interferometer: On the left
           the input state \textcolor{OliveGreen}{$\rho$} (in green), representing the state
           of a particle on $d$ paths. Each path can be subjected to 
           a local relative phase \textcolor{magenta}{$\alpha_j$} (in magenta), after which
           the particle is detected via a suitable measurement (in violet)
           that produces an outcome \textcolor{red}{$\omega$} (in red).}
  \label{fig:interferometer}
\end{figure}
For the moment, we will ignore any internal degrees of freedom of the
particle, and any other spatial degrees,
so that the entire Hilbert space describing the system is $\cH$, and
a pure state inside the interferometer can be written as 
$\ket{\psi} = \sum_j c_j \ket{j}$, and a general mixed state as
\begin{equation}
  \rho = \sum_{j,k=1}^d \rho_{jk} \ketbra{j}{k}.
\end{equation}
An interferometric experiment (Fig. \ref{fig:interferometer}) has two distinct components. 
The first consists of local phase shifts $\alpha_j$ that can be inserted
into the paths, implementing a diagonal phase unitary
\begin{equation}
  U(\vec{\alpha}) = \sum_j e^{i\alpha_j} \proj{j},
\end{equation}
so that the state becomes
\begin{equation}
  \rho(\vec{\alpha}) = U(\vec{\alpha}) \rho U(\vec{\alpha})^\dagger
                     = \sum_{j,k=1}^d e^{i(\alpha_j-\alpha_k)}\rho_{jk} \ketbra{j}{k}.
\end{equation}
The second is a detector at the output, often simply fixed as the
combination of a symmetric beam splitter with a path measurement,
but for us a general POVM $M=(M_\omega)$, with outcomes $\omega$
from a suitable space $\Omega$.

The experimenter, having chosen $\vec{\alpha} = (\alpha_1,\ldots,\alpha_d)$,
will observe outcomes $\omega\in\Omega$ sampled from the Born distribution
(the ``interference pattern''):
\begin{equation}
  \label{eq:response}
  P_{M|\rho}(\omega|\vec{\alpha}) = \tr U(\vec{\alpha}) \rho U(\vec{\alpha})^\dagger M_\omega.
\end{equation}
The signature of interference in such an experiment, where $\rho$
is given and fixed, is that the distribution $P=P_{M|\rho}$ can vary as
a function of the phases $\alpha_j$. The degree of variability,
intuitively, is the visibility of the interference pattern, calling
for a visibility functional $V=V[P]$ on conditional 
distributions $P(\omega|\vec{\alpha})$.

While it is always dangerous to make a priori demands, we take it that 
such a functional has to capture the global property of
$P$ not being constant. I.e., it should be $0$ for constant 
$P(\cdot|\vec{\alpha})$ and positive otherwise. 
It will also make sense to ask that it is 
invariant under permutations and shifts of the $\vec{\alpha}$, reflecting the
obvious symmetries of the experimental setup (see Fig.~\ref{fig:interferometer}).
We call a visibility functional $V[P]$ satisfying these requirements
\emph{regular}.
In the discussion of interferometers, specifically in the rich literature on
the complementary between fringe visibility and which-path information
\cite{WoottersZurek,GreenbergerYasin,Englert,Jaeger-et-al}, 
to cite only the principal ones, the topic of visibility has been
addressed repeatedly, and from increasingly general perspectives.
In particular, the realisation that for $d > 2$ no unique visibility functional 
seems to exist, called for an axiomatic approach to put order into the many ad hoc
parameters; cf.~\cite{Duerr,EKKC,Prillwitz-et-al}. We wish to highlight especially 
Coles' paper~\cite{Coles:coherent}, which makes an eloquent case for an \emph{operational} 
approach, where visibility (as well as which-path information) is expressed as 
a property of an observable probability distribution, and with which philosophy
we feel very much in line.

\medskip
In the simplest case of the well-known Mach-Zehnder interferometer (Fig.~\ref{fig:MZ}),
i.e.~$d=2$, we observe interference fringes w.r.t.~a relative phase shift:
Consider the density matrix
$\rho = \rho_{11} \proj{1} + \rho_{22} \proj{2} 
         + \rho_{12} \ketbra{1}{2} + \rho_{21} \ketbra{2}{1}$,
the diagonal phase unitary
$U(\vec{\alpha}) = e^{i\alpha_1} \proj{1} + e^{i\alpha_2} \proj{2}$,
and a measurement with POVM elements $\proj{\mu}$, 
$\ket{\mu} = \mu_1\ket{1}+\mu_2\ket{2}$.
Now, with $\alpha = \alpha_1-\alpha_2$, and writing
$\rho_{12}=|\rho_{12}|e^{i\beta}$, $\overline{\mu_1}\mu_2 = |\mu_1\mu_2|e^{i\gamma}$, 
the output probability is 
\begin{equation}\begin{split}
  P_{M|\rho}(\mu|{\vec{\alpha}})
     &= \rho_{11}|\mu_{1}|^2 + \rho_{22}|\mu_{2}|^2 \\
     &\phantom{===}
        + 2|\rho_{12}\mu_1\mu_2|\cos(\alpha+\beta+\gamma),
\end{split}\end{equation}
whose fluctuation is essentially characterised by 
the coefficient $|\rho_{12}\mu_1\mu_2|$, and so most
analyses conclude this to be the visibility.

\begin{figure}[ht]
	\includegraphics[width=0.48\textwidth]{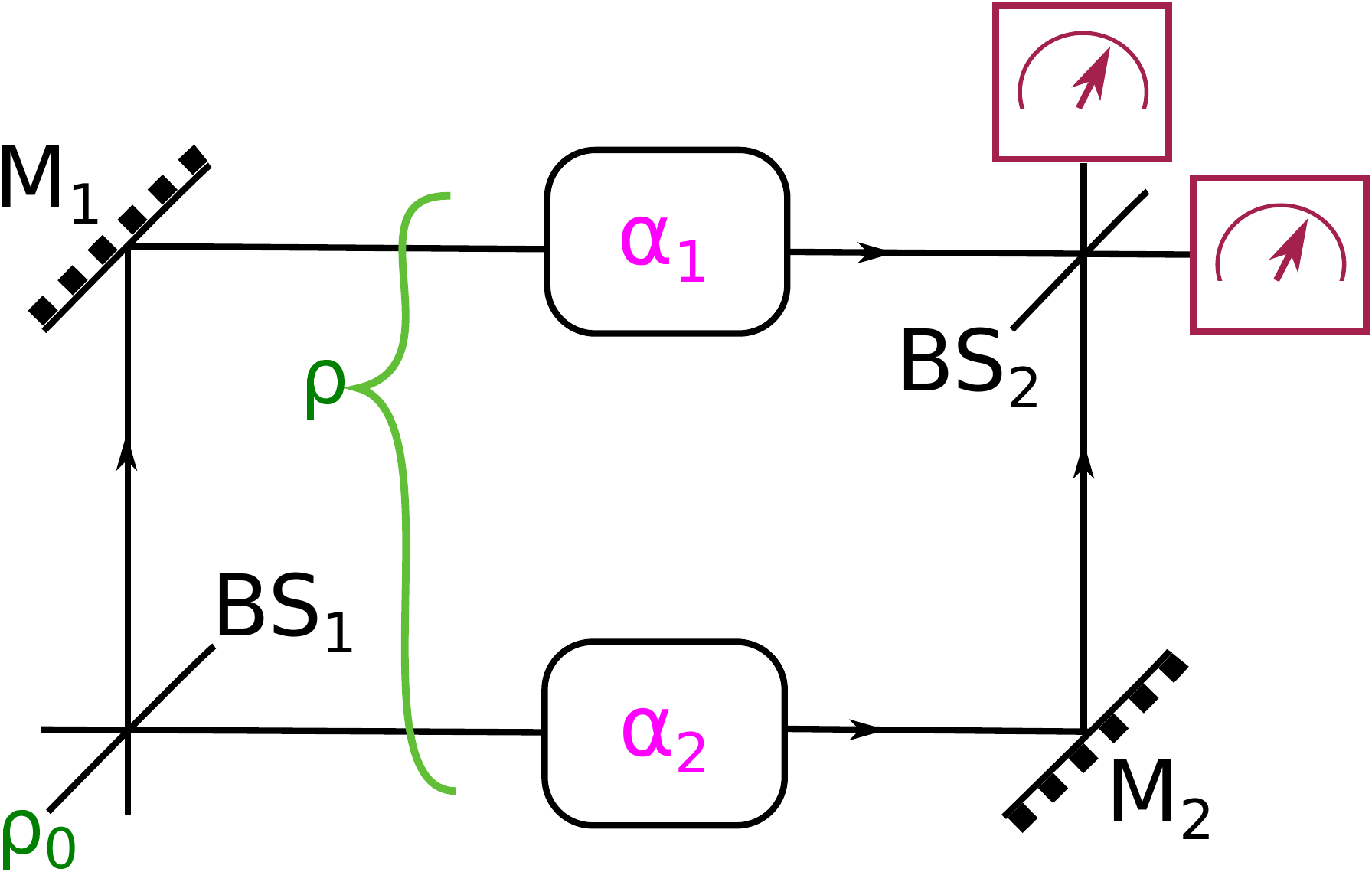}
	\caption{Mach-Zehnder two-path interferometer. The initial state 
	         is $\rho_0$, while the state after the interaction with the first beam splitter (BS1) 
	         and the first mirror (M1) is \textcolor{green}{$\rho$} (in green), 
	         representing the state of a particle on two paths. Each path can be 
	         subjected to a phase {$\alpha_1$} and {$\alpha_2$}, although only their
	         relative difference $\alpha = \alpha_1-\alpha_2$ is physically relevant. 
	         After the interaction with 
	         the second beam splitter (BS2), the interference pattern is observed.}
	\label{fig:MZ}
\end{figure}


\bigskip
\textbf{Optimal visibility as measure of coherence.}---%
If we want to treat the state $\rho$ as a resource, i.e.~as a given,
of which we are supposed to make the best, it makes sense to optimise
the visibility $V[P_{M|\rho}]$ over all possible measurements. 
This idea diverges somewhat from lab practice in interferometry and from
many discussion of visibility vs.~which-path information duality, where 
a fixed measurement is used, usually one mixing the paths uniformly
(a $50$-$50$ beam splitter in the Mach-Zehnder case, and more
generally a transformation acting as a Fourier transform, followed
by a channel detector)~\cite{Englert,EKKC,BimonteMusto,KszliZuk}. 

However, it is clear that the most beautiful coherent superposition
in the state may be rendered invisible by an unsuitable choice of measurement
For instance, consider the qutrit state
\(
  \rho = \frac13\left(\ket{1}\!+\!\ket{2}\right)\!\left(\bra{1}\!+\!\bra{2}\right) + \frac13 \proj{3},
\)
under a measurement $M^{(0)}$ in the basis 
$\left\{ \ket{1}, \frac{1}{\sqrt{2}}(\ket{2}\pm\ket{3}) \right\}$;
evidently, the three outcomes all have probability $\frac13$, 
irrespective of the phases in
\[
  \rho(\alpha)
    = \frac13 \left(\ket{1}\!+\!e^{i\alpha}\ket{2}\right)\!\left(\bra{1}\!+\!e^{-i\alpha}\bra{2}\right)
       +\frac13 \proj{3}.
\]
Intuitively, we expect the best choice to bring out the coherence
in $\rho$ to be the projective measurement $M^{(1)}$ in the basis
$\left\{ \frac{1}{\sqrt{2}}(\ket{1}\pm\ket{2}), \ket{3} \right\}$,
for which the detection probabilities are
$(\frac23 \cos^2\frac{\alpha}{2}, \frac23 \sin^2\frac{\alpha}{2}, \frac13)$.
On the other hand, the standard choice of a symmetric beam splitter 
results in the Fourier basis 
$\left\{ \frac{1}{\sqrt{3}}(\ket{1}+\zeta^t\ket{2}+\zeta^{2t}\ket{3}) : t=0,1,2 \right\}$,
where $\zeta = e^{2\pi i/3}$, with detection probabilities
\[
  \Bigl(\frac19 + \frac49 \cos^2(\frac{\alpha}{2}-\frac{t\pi}{3}) : t=0,1,2 \Bigr),
\]
which has the same oscillation pattern as $M^{(1)}$, but smaller
amplitude.

Thus we are motivated, given a visibility functional $V[P]$, to optimise
the visibility over all measurements, to get the best out of $\rho$.
This leads to a number that now depends only on the state,
\begin{equation}
  \label{eq:CV}
  C_V(\rho) := \sup_{M=(M_\omega)} V[P_{M|\rho}].
\end{equation}

The hypothesis that we will explore in the rest of the paper is that
this number, for a large class of visibility functionals, is a good
indicator of coherence in $\rho$.

In an attempt to identify consistent quantifiers
of coherent superposition, Baumgratz \emph{et al.}~\cite{BCP14}, following
\AA{}berg~\cite{Aaberg2006}, have created a resource theory of
coherence, with carefully chosen resource free state (diagonal
density operators $\Delta$) and free transformations (the so-called \emph{incoherent
operations (IO)}). A measure that is non-increasing under these operations
is called a \emph{monotone}.
Mathematically, incoherent operations are completely positive
and trace preserving linear maps, built from \emph{incoherent Kraus operators},
\(
  T(\rho) = \sum_\lambda K_\lambda \rho K_\lambda^\dagger,
\)
where $K_\lambda\ket{j} \propto \ket{k} = \ket{k(\lambda,j)}$ for all computational
basis elements $\ket{j}$. If $K_\lambda$ and $K_\lambda^\dagger$ are both
incoherent, it is called \emph{strictly incoherent (SIO)}, and so is a map $T$ built
from strictly incoherent Kraus operators. From this is it straightforward to
see that a Kraus operator is incoherent if and only if it has the form 
$K = \sum_{j} c_j \ketbra{k(j)}{j}$ with a function $k(j)$ mapping basis
states to basis states; it is strictly incoherent if and only if $k$
is one-to-one.

A functional $C(\rho)\geq 0$ is a monotone if $C(\rho) \geq C(T(\rho))$
under all IO (SIO) $T$. It is called a \emph{strong monotone} if for an incoherent
Kraus decomposition, $K_\lambda \rho K_\lambda^\dagger =: q_\lambda \rho_\lambda$,
$C(\rho) \geq \sum_\lambda q_\lambda C(\rho_\lambda)$. Well-known examples
include the $\ell_1$-measure of coherence \cite{BCP14}, and the relative entropy
of coherence \cite{Aaberg2006,BCP14,WinterYang}:
\begin{align}
  C_{\ell_1}(\rho) &= \sum_{j\neq k} |\rho_{jk}|, \\
  C_r(\rho)        &= \min_{\sigma\in\Delta} D(\rho\|\sigma) = S(\Delta(\rho))-S(\rho),
\end{align}
with the relative entropy $D(\rho\|\sigma) = \tr\rho(\log\rho-\log\sigma)$ 
and the von Neumann entropy $S(\rho) = -\tr\rho\log\rho$;
$\Delta(\rho) = \sum_j \proj{j} \rho \proj{j}$ is the diagonal part of $\rho$.

Our first result is a general link between visibility and coherence.
We call a visibility functional \emph{weakly affine} if for distributions
$P_i(\omega|\vec{\alpha})$, $\omega \in \Omega_i$ (assuming 
w.l.o.g.~pairwise disjoint $\Omega_i$), and for a probability
distribution $(q_i)$, we have
$V\left[\overline{P}\right] = \sum_i q_i V[P_i]$, with the
averaged distribution $\overline{P} = \sum_i q_i P_i$ on 
$\Omega = \bigcup_i \Omega_i$. 

\begin{theorem}
  \label{thm:C_V:monotone}
  For any regular and weakly affine visibility functional $V[P]$,
  $C_V$ is a coherence measure that is strongly monotonic under
  SIO (strictly incoherent operations).
  If $V$ is convex in $P$, then $C_V$ is convex in $\rho$.
\end{theorem}
\begin{proof}
Let a SIO with Kraus operators $K_\lambda$ be given, acting on a
state $\rho$, so that $q_\lambda \rho_\lambda = K_\lambda \rho K_\lambda^\dagger$
defines the probability of the event $\lambda$ and the post-measurement state.
Observe that, because $K_\lambda = \pi_\lambda D_\lambda$ can be written as a 
diagonal matrix $D_\lambda$ followed by a permutation $\pi_\lambda$,
\[\begin{split}
  q_\lambda U(\vec{\alpha}) \rho_\lambda U(\vec{\alpha})^\dagger
      &= U(\vec{\alpha}) K_\lambda \rho K_\lambda^\dagger U(\vec{\alpha})^\dagger \\
      &= K_\lambda U(\vec{\beta}) \rho U(\vec{\beta})^\dagger K_\lambda^\dagger,
\end{split}\]
with $\beta_j = \alpha_{\pi_\lambda(j)}$. This shows that the probability
of seeing outcome $\lambda$ is $q_\lambda$ for all $\rho(\vec{\alpha})$.

Now choose measurements $M^{(\lambda)}$ for each $\rho_\lambda$,
taking values $\omega$ in the disjoint sets $\Omega_\lambda$, subject to the 
probability law $P_\lambda = P_{M^{(\lambda)}|\rho_\lambda}$ given by
\[\begin{split}
  P_\lambda(\omega|\vec{\alpha}) 
     &= \tr U(\vec{\alpha}) \rho_\lambda U(\vec{\alpha})^\dagger M^{(\lambda)}_\omega \\
     &= \frac{1}{q_\lambda} 
         \tr U(\vec{\alpha}) K_\lambda \rho K_\lambda^\dagger U(\vec{\alpha})^\dagger M^{(\lambda)}_\omega \\
     &= \frac{1}{q_\lambda}
         \tr K_\lambda U(\vec{\beta}) \rho U(\vec{\beta})^\dagger K_\lambda^\dagger M^{(\lambda)}_\omega \\
     &= \frac{1}{q_\lambda}
         \tr U(\vec{\beta}) \rho U(\vec{\beta})^\dagger K_\lambda^\dagger M^{(\lambda)}_\omega K_\lambda.
\end{split}\]
Introducing the POVM 
$\widetilde{M} = (K_\lambda^\dagger M^{(\lambda)}_\omega K_\lambda)_{\lambda,\omega}$
with outcomes $(\lambda,\omega)$, we can now invoke weak affinity:
\[
  \sum_\lambda q_\lambda V[P_\lambda] =    V\!\left[ \sum_\lambda q_\lambda P_\lambda \right] 
                                      =    V[P_{\widetilde{M}|\rho}] 
                                      \leq C_V(\rho),
\]
because the measurement $\widetilde{M}$ is eligible for $\rho$ but may be
suboptimal. Since the measurements $M^{(\lambda)}$ can be chosen to maximise 
the left hand side, we obtain 
\[
  \sum_\lambda q_\lambda C_V(\rho_\lambda) \leq C_V(\rho).
\]

For the convexity statement, let $\rho = \sum_i p_i \sigma_i$ and choose any
measurement $M$ on $\rho$. Then,
\[
  V[P_{M|\rho}] =    V\!\left[ \sum_i p_i P_{M|\sigma_i} \right] 
                \leq \sum_i p_i V[P_{M|\sigma_i}] 
                 \leq \sum_i p_i C_V(\sigma_i),
\]
and because $M$ may be chosen to maximise the left hand side, we find
\(
  C_V\!\left(\sum_i p_i \sigma_i\right) \leq \sum_i p_i C_V(\sigma_i),
\)
as claimed.
\end{proof}

\medskip
A couple of remarks are in order: 
First, there don't seem to be easy conditions for $C_V$ to be a (strong) 
coherence monotone under IO, but of course that is something potentially checkable
in individual cases.
Secondly, one might wonder in case we merely want to \emph{detect} coherence, 
whether there is a universal measurement $M$ such that if $\rho$ has coherences, 
then $V[P_{M|\rho}]$ is positive. The answer is yes, namely any tomographically
complete measurement, as long as $V[P]$ has the property that it is non-zero
on every non-constant $P$.

\bigskip
\textbf{Examples.}---%
We now show that the above theory is not just an abstract
construction, by considering several concrete visibility parameters,
for which we can evaluate the associated coherence measures, or at least
considerably simplify the optimisation.

\medskip
\textit{A. Largest difference of intensity:}
The perhaps simplest and most intuitive parameter of visibility for
two-outcome measurements $M=(M_0,M_1=\1-M_0)$ is the difference between
the largest and the smallest value of 
$P_{M|\rho}(0|\vec{\alpha}) = \tr U(\vec{\alpha}) \rho U(\vec{\alpha})^\dagger M_0$.
To make it suitable for measurements with arbitrary outcome sets, we define
\begin{equation}
  V_{\max}[P] := \sup_{\vec{\alpha},\vec{\beta}} 
                   \frac12 \left\| P(\cdot|\vec{\alpha})-P(\cdot|\vec{\beta}) \right\|_1.
\end{equation}
Note that we do not normalise by the sum of the largest and smallest
probability, as is customary in discussions of visibility in classical 
interferometry, where the basic observable quantities are intensities.
There, this appears necessary to obtain a dimensionless visibility;
here however, we have the probabilities that are already dimensionless
and have an absolute meaning.

Clearly, $V_{\max}$ is regular and weakly affine, so the 
corresponding coherence measure $C_{\max}$ is a SIO monotone.
In fact, it is easy to evaluate it, and the result is
\begin{equation}\begin{split}
  \label{eq:C_max}
  C_{\max}(\rho) &= \max_{\vec{\alpha}} 
                      \frac12 \left\| U(\vec{\alpha}) \rho U(\vec{\alpha})^\dagger - \rho \right\|_1 \\
                 &= \max_{\vec{\alpha}} 
                      \frac12 \bigl\| [\rho, U(\vec{\alpha})] \bigr\|_1 \\
                 &= \max_{\vec{\alpha},\,M_0} 
                      \tr U(\vec{\alpha}) \rho U(\vec{\alpha})^\dagger M_0 - \tr \rho M_0,
\end{split}\end{equation}
because we can always shift $\vec{\beta}$ to $\vec{0}$ by applying $U(-\vec{\beta})$.
In particular, the optimal measurement is a two-outcome POVM $(M_0,M_1=\1-M_0)$, and
the value is the largest difference in response probability over POVM elements.

We can compare the result with the trace distance measure of coherence,
$C_{\tr}(\rho) = \min_{\sigma\in\Delta} \frac12 \| \rho-\sigma \|_1$,
introduced in~\cite{BCP14}:
\(
  C_{\tr}(\rho) \leq C_{\max}(\rho) \leq 2 C_{\tr}(\rho).
\)

Namely, on the one hand, for $\sigma \in \Delta$, we have
$\| \rho-\sigma \|_1 = \left\| U(\vec{\alpha})\rho U(\vec{\alpha})^\dagger - \sigma \right\|_1$,
so by the triangle inequality
\[
  \left\| U(\vec{\alpha})\rho U(\vec{\alpha})^\dagger - \rho \right\|_1
     \leq \left\| U(\vec{\alpha})\rho U(\vec{\alpha})^\dagger - \sigma \right\|_1
             + \| \rho - \sigma \|_1,
\]
which implies $C_{\max}(\rho) \leq 2 C_{\tr}(\rho)$. On the other hand,
\[\begin{split}
  C_{\tr}(\rho) &\leq \frac12 \bigl\| \rho-\Delta(\rho) \bigr\|_1 \\
                &= \frac12 \left\| \rho 
                         - \int \!{\rm d}\vec{\alpha}\, 
                                U(\vec{\alpha})\rho U(\vec{\alpha})^\dagger \right\|_1 \\
                &\leq \int \!{\rm d}\vec{\alpha}\, 
                           \frac12 \left\| U(\vec{\alpha})\rho U(\vec{\alpha})^\dagger - \rho \right\|_1 
                 \leq C_{\max}(\rho).
\end{split}\]

\medskip
In the qubit case, it holds that 
$C_{\max}(\rho) = 2|\rho_{01}| = C_{\ell_1}(\rho) = 2 C_{\tr}(\rho)$
(see Appendix).

\medskip
\textit{B. Estimating equidistributed phases:}
Inspired by the previous example, we are motivated to considering
guessing problems of a more general kind, where we are trying to
estimate the true setting of the phases among several alternatives,
based on measurement outcomes. It turns out that a good
candidate is the equidistributed set of $d$ phases $\frac{2\pi j}{d}(1,2,\ldots,d)$,
$j=1,\ldots,d$, and its shifts and permutations:
\begin{equation}
  V_{\text{guess}}[P] := -\frac1d + \max_{\substack{\vec{\alpha}_0,\,\pi\in S_d \\ 
                                          \Omega_j \cap \Omega_k = \emptyset}}
                                     \frac1d \sum_{j=1}^d P(\Omega_j|\vec{\alpha}_0+j\vec{h}_\pi),
\end{equation}
where $\vec{h}_\pi = \frac{2\pi}{d}\bigl(\pi(1),\pi(2),\ldots,\pi(d)\bigr)$
is a generating vector of uniformly accelerating phases (w.r.t.~the permutation
$\pi$ of coordinates).
This quantity is the bias (excess over $\frac1d$) of the optimal strategy
to guess the true value of $j \in \{1,\ldots,d\}$ that defines the phase
settings.
As defined, this visibility functional is regular and weakly affine,
so the corresponding $C_{\text{guess}}$ is a coherence monotone under SIO.
As a matter of fact, it holds \cite{robustness}
\begin{equation}
  C_{\text{guess}}(\rho) = -\frac1d + \max_{(M_j) \text{ POVM}}
                                       \frac1d \sum_{j=1}^d 
                                             \tr \rho(\vec{\alpha}_0+j\vec{h}_\pi) M_j 
                         = \frac1d C_R(\rho),
\end{equation}
for any $\vec{\alpha}_0$ and any permutation $\pi$. 
Here, $C_R$ denotes the \emph{robustness of coherence}, defined via
\begin{equation}
  1 + C_R(\rho) = \min \tr \delta \text{ s.t. } \delta \geq \rho,\ \delta\in\Delta,
\end{equation}
which is known to be an IO monotone \cite{robustness}.
Interestingly, by maximising the operational visibility proposed
in~\cite[Eq.~(20)]{Coles:coherent}, the same result is obtained \cite{Coles:unpub}.

\medskip
In the qubit case, it is well known that the robustness of coherence
equals the $\ell_1$-measure:
$C_R(\rho) = 2|\rho_{01}| = C_{\ell_1}(\rho)$ \cite{robustness},
and so $C_{\text{guess}}(\rho) = |\rho_{01}|$ is just half of that.

\medskip
\textit{C. Largest sensitivity to phase changes:}
Looking back at example A, we notice that the points of largest and
smallest value of the response probability
$I(\vec{\alpha}) = P_{M|\rho}(0|\vec{\alpha}) 
                 = \tr U(\vec{\alpha}) \rho U(\vec{\alpha})^\dagger M_0$
to a POVM element $M_0$ may be quite far apart. In contrast, 
in many applications of interferometry it is a relatively small phase difference
that we want to pick up \cite{LIGO}, 
so we are interested in the largest magnitude of the derivative of $I(\vec{\alpha})$:

\begin{equation}
  V_{\nabla}[P] := \max_{\vec{\alpha},\vec{h}} 
                      \left| \frac{\partial I}{\partial \vec{h}}(\vec{\alpha}) \right|,
\end{equation}
where $\vec{\alpha}$ ranges over all phases, and $\vec{h}$ over all direction
vectors that are suitably norm bounded. 
To extend $V_{\nabla}$ to general measurements, we may include
a maximisation over all two-outcome coarse grainings.
We can easily see that $V_{\nabla}[P]$ is regular and weakly affine since 
$I(\vec{\alpha})$ is well-defined probability distribution over $\vec{\alpha}$.

Now, as $I(\vec{\alpha}) = \tr \rho U(\vec{\alpha})^\dagger M_0 U(\vec{\alpha})$, 
its derivative at (w.l.o.g.) $\vec{0}$ in direction $\vec{h}$ is given by
\begin{equation}
  \label{eq:commutator}
  \frac{\partial I}{\partial \vec{h}}(\vec{0}) = - i \tr [\rho,H] M_0
                                               = - i \tr \rho [H,M_0],
\end{equation}
where $H$ is the diagonal Hamiltonian with eigenvalues $h_j$,
$H=\operatorname{diag}(\vec{h})$. 
Note that the derivative at any other point $\vec{\alpha}_0$ is the same, 
up to conjugating the measurement by $U(\vec{\alpha}_0)$.
There are two natural limitations on
$\vec{h}$: Geometrically, to obtain the largest gradient of $I$,
we should consider unit vectors $\vec{h}$, meaning $\| H \|_2^2 = \tr H^2 = 1$;
or taking motivation from the Hamiltonian, we should bound its energy 
range, meaning $\| H \|_\infty \leq 1$. We denote these two scenarios
by $p=2$ and $\infty$, giving rise to two coherence measures
$C^{(p)}_{\nabla}$.
From Eq.~(\ref{eq:commutator}), we directly get
\begin{equation}
  \label{eq:nabla}
  C^{(p)}_{\nabla}(\rho) 
        = \max \frac12 \bigl\| [\rho,H] \bigr\|_1 \text{ s.t. } H \text{ diag.}, \|H\|_p \leq 1.
\end{equation}
Inspecting this formula, we see that the
optimisation is convex in $H$, hence the maximum is attained on an
extremal admissible Hamiltonian. 
For $p=2$, these have the form
$H=\sum_j \epsilon_j \sqrt{t_j} \proj{j}$, with $\epsilon_j=\pm 1$
and $\sum_j t_j = 1$. 
For $p=\infty$, the extremal $H$ have entries $\pm 1$ along the diagonal, and so
\begin{equation}
  C^{(\infty)}_{\nabla}(\rho) 
        = \max_{S_+ \stackrel{.}{\cup} S_- = [d]} 2\, \| \Pi_+ \rho \Pi_- \|_1,
\end{equation}
where the maximisation over partitions $S_+ \stackrel{.}{\cup}S_- = [d]$, 
with $\Pi_{\bullet} = \sum_{j\in S_{\bullet}} \proj{j}$, $\bullet=\pm$.
In both cases, we obtain a strong SIO monotone, due to the evident 
weak affinity of $V_{\nabla}$. From Eq.~(\ref{eq:C_max}) we see that
$C^{(\infty)}_{\nabla} \leq C_{\max}$, but equality does
not seem to hold in general.

An alternate form of $C_{\nabla}$ can be obtained by
using Eq.~(\ref{eq:commutator}), and 
going to the more convenient variable $B=2M_0-\1$ in the above
equations. After a few manipulations, we arrive at 
\begin{equation}
  C^{(p)}_{\nabla}(\rho) = \max \tr \rho X \text{ s.t. } X \in \mathcal{C}_p,
\end{equation}
where the maximisation is over the set of Hermitian matrices
\[
  \mathcal{C}_p = \biggl\{ X=\frac{1}{2i}[H,B] : H=H^\dagger \text{ diagonal},\ 
                                \| H \|_p \leq 1,\, -\1 \leq B  \leq \1 \biggr\}.
\]
In this form it is formally a convex optimisation problem, because
we may as well go to the convex hull of $\mathcal{C}_p$. However,
its characterisation remains as a beautiful open problem.
Indeed, it is easy to see that the elements of $\mathcal{C}_p$ have zero diagonal
and satisfy $\| X \|_p = \bigl( \tr|X|^p \bigr)^{1/p} \leq 1$, but there may
be other constraints.

\medskip
Once again, the qubit case is very simple (see Appendix): 
$C^{(2)}_{\nabla}(\rho) = \sqrt{2}|\rho_{12}|$ and 
$C^{(\infty)}_{\nabla}(\rho) = 2|\rho_{12}| = C_{\ell_1}(\rho)$.

\medskip
\textit{D. Largest Fisher information:}
Considering further the previous example, we realise that finding the
largest derivative of the probability $P(0|\vec{\alpha})$, while 
strongly motivated by the intuition rooted in intensities, does not 
necessarily identify the point of strongest statistical sensitivity,
which is asking for the largest Fisher information, the natural measure
for probability distributions. Looking again at directional estimation 
of a one-dimensional subfamily $\vec{\alpha} = t\vec{h} + \vec{\alpha}_0$, $t\in\mathbb{R}$,
the Fisher information is given by the expected squared logarithmic
derivative of the probability distribution:
\begin{equation}\begin{split}
  \mathcal{F}_{\vec{\alpha}_0}(\vec{h}) 
              &= \sum_{\omega\in\Omega} P(\omega|\vec{\alpha}_0) 
                    \left(\frac{{\rm d}\ln P(\omega|\vec{\alpha})}{{\rm d}t} \Big\vert_{t=0} \right)^2 \\
              &= \sum_{\omega\in\Omega} \frac{1}{P(\omega|\vec{\alpha}_0)} 
                     \left(\frac{{\rm d} P(\omega|\vec{\alpha})}{{\rm d}t} \Big\vert_{t=0} \right)^2 \!,
\end{split}\end{equation}
so we are considering the visibility functional
\begin{equation}
  V_{\text{F}}[P] := \max_{\vec{\alpha}_0,\vec{h}} \mathcal{F}_{\vec{\alpha}_0}(\vec{h}),
\end{equation}
where $\vec{\alpha}_0$ varies over the whole space of phases, and $\vec{h}$
over a suitably bounded set of directions. 
Clearly, $V_{\text{F}}$ is regular and weakly affine.

The formula for the Fisher information, optimised over measurements
(and $\vec{\alpha}_0$, which w.l.o.g.~is $\vec{0}$, by the same reasoning as
in previous examples),
for estimating $t\approx 0$ in $e^{-itH}\rho e^{itH}$ for a given
diagonal Hamiltonian $H = \operatorname{diag}(\vec{h})$ 
and $\rho = \sum_j \lambda_j \proj{e_j}$ is known \cite{BraunsteinCaves,Paris:Fisher} 
and given by
\begin{equation}
  \label{eq:Fisher}
  \mathcal{F}_{\text{opt}}(\vec{h})
     = 2 \sum_{jk} \frac{(\lambda_j\!-\!\lambda_k)^2}{\lambda_j+\lambda_k}
                   |\bra{e_j} H \ket{e_k}|^2. 
\end{equation}

Like in the previous example on sensitivity, there are two natural domains
of diagonal Hamiltonians $H$ over which to optimise this: 
Either $\| H \|_2 \leq 1$
or $\| H \|_\infty \leq 1$,
leading to two variants $C^{(2)}_{\text{F}}(\rho)$ and $C^{(\infty)}_{\text{F}}(\rho)$
of the coherence measure.

In either case, the optimal choice of $H$ is extremal subject to the convex
constraint, because $\mathcal{F}$ can easily be seen to be convex in $H$.
Namely, each term $|\bra{e_j} H \ket{e_k}|$ is convex, hence also its square,
and the coefficient in front of it manifestly nonnegative.
Thus, we obtain:
\begin{align}
  C^{(2)}_{\text{F}}(\rho)      
          &= \max_{\substack{\sum_j t_j = 1 \\ \epsilon_j=\pm 1}}
              \phantom{:} \sum_{jk}\,
                       2\frac{(\lambda_j\!-\!\lambda_k)^2}{\lambda_j+\lambda_k}  \nonumber\\[-2mm]
          &\phantom{========:}
               \left|\bra{e_j}\!\!\left(\!\sum_j \!\epsilon_j\!\sqrt{t_j} \proj{j}\!\right)\!\!
                                                                    \ket{e_k}\right|^2\!\!\!,\\
  C^{(\infty)}_{\text{F}}(\rho) 
          &= \max_{S_+,S_- \subset [d]} 
              \sum_{jk} 2\frac{(\lambda_j\!-\!\lambda_k)^2}{\lambda_j+\lambda_k}  
            \bigl|\bra{e_j} (\Pi_+ \!-\! \Pi_-) \ket{e_k}\bigr|^2\!\!,
\end{align}
where the first maximisation is over diagonal Hamiltonians with Hilbert-Schmidt
norm $1$; the second over partitions $S_+ \stackrel{.}{\cup}S_- = [d]$, 
with $\Pi_{\bullet} = \sum_{j\in S_{\bullet}} \proj{j}$, $\bullet=\pm$,
so that $H = \Pi_+ - \Pi_-$.

\medskip
For a qubit state $\rho$, it can be verified (see Appendix) that
$C^{(2)}_{\text{F}}(\rho) = 4|\rho_{12}|^2 = C_{\ell_1}(\rho)^2$, 
$C^{(\infty)}_{\text{F}}(\rho) 
 = 2 C_{\ell_1}(\rho)^2$.

\medskip
\textit{E. Largest differential Chernoff bound:}
We observe that the attainability of the Fisher information presupposes 
access to many copies of the state and independent measurements,
in which setting the Fisher information gives the optimal scaling
of the mean squared estimation error with the number of copies. 
If we allow general collective measurements and at the same time
only want to distinguish pairs of nearby states optimally, we are led 
to the differential Chernoff bound \cite{Calsamiglia-et-al:Chernoff}:
While the Chernoff bound is defined as 
$\xi(\rho,\sigma) = \sup_{0\leq s \leq 1} - \ln \tr \rho^s\sigma^{1-s}$,
for states and probability distributions alike \cite{Calsamiglia-et-al:Chernoff,Chernoff}, 
it is known that 
\(
\frac{1}{{\rm d}t^2}{\rm d}^2\xi\bigl( P(\cdot|\vec{\alpha}_0),P(\cdot|\vec{\alpha}_0+{\rm d}t\vec{h})\bigr)
     =: {\rm d}_{\vec{h}}\xi^2
\)
defines the line element of a Riemannian metric on the parameter space. 
Thus we let
\begin{equation}
  V_{\partial\xi}[P] := \max_{\vec{\alpha}_0,\vec{h}} {\rm d}_{\vec{h}}\xi^2.
\end{equation}
As $V_{\partial\xi}$ is regular and weakly affine, we will obtain a strong
SIO monotone. Note that this would not work simply fixing a Hamiltonian,
as shown in \cite{Girolami:wrong,Luo:WY}.

The differential Chernoff bound, optimised over measurements,
for distinguishing $e^{-itH} \rho e^{itH}$ for $t\approx 0$ from $\rho$
in the many-copy regime,
with a diagonal Hamiltonian $H$ and $\rho = \sum_j \lambda_j \proj{e_j}$ is
again known \cite{Calsamiglia-et-al:Chernoff}, and given by
${\rm d}_H\xi^2 = \frac{1}{{\rm d}t^2}{\rm d}^2\xi\bigl( \rho,e^{-itH}\rho e^{itH} \bigr)$,
which evaluates to 
\begin{equation}\begin{split}
  \label{eq:diff-Chernoff}
  {\rm d}_H\xi^2
        &= \frac12 \sum_{jk} \left(\sqrt{\lambda_j}-\sqrt{\lambda_k}\right)^2 |\bra{e_j} H \ket{e_k}|^2 \\
        &= \frac12 \sum_{jk} \left( \lambda_j+\lambda_k-2\sqrt{\lambda_j\lambda_k}\right) 
                                                                            |\bra{e_j} H \ket{e_k}|^2 \\
        &= \tr \rho H^2 - \tr \sqrt{\rho} H \sqrt{\rho} H
         = - \frac12 \tr [\sqrt{\rho},H]^2,
\end{split}\end{equation}
the latter equalling the \emph{Wigner-Yanase skew information}, 
$I_{\text{WY}}(\rho,H)$ \cite{skew-info}.

Like in the previous two examples, there are two natural domains
of diagonal Hamiltonians $H$ over which to optimise this: 
Either $\| H \|_2 \leq 1$ or $\| H \|_\infty \leq 1$,
leading to two variants $C^{(2)}_{\partial\xi}(\rho)$ and $C^{(\infty)}_{\partial\xi}(\rho)$
of the coherence measure.
Again, ${\rm d}_H\xi^2$ is convex in $H$, thanks to convexity
of each term $|\bra{e_j} H \ket{e_k}|^2$,
and $\left(\sqrt{\lambda_j}-\sqrt{\lambda_k}\right)^2 \geq 0$.
Consequently, the optimal $H$ is extremal under the convex norm constraint.
For $p=\infty$, this means that the maximum is attained on a difference
of two diagonal projectors, $H=\Pi_+-\Pi_-$. For $p=2$,
however, we can say something even better, using
Lieb's concavity theorem \cite{Lieb:WYD-concavity}, which
says that for semidefinite $H$, the Wigner-Yanase skew information is 
convex in $H^2$, by writing $H = \sqrt{H^2}$.
In general, we split $H=H_+-H_-$ into positive and negative
parts, and find after some straightforward algebra that
\[
  I_{\text{WY}}(\rho,H) = I_{\text{WY}}(\rho,H_+) + I_{\text{WY}}(\rho,H_-)
                                          - 2 \tr \sqrt{\rho}H_+\sqrt{\rho}H_-,
\]
which by Lieb's theorem \cite{Lieb:WYD-concavity}
is jointly convex in $H_+^2$ and $H_-^2$. 
Thus we find that the optimal $H_+$ and $H_-$ must be proportional to 
rank-one projectors, resulting in the expression claimed for $C_{\partial\xi}^{(2)}(\rho)$.

\begin{align}
  C^{(2)}_{\partial\xi}(\rho) 
           &= \max_{j,k,t} \ 
                I_{\text{WY}}\!\left(\rho,\sqrt{t}\proj{j}-\sqrt{1\!-\!t}\proj{k}\right)\! , \\
  C^{(\infty)}_{\partial\xi}(\rho) 
           &= \max_{S_+,S_- \subset [d]} I_{\text{WY}}\bigl(\rho,\Pi_+-\Pi_-\bigr)  \nonumber \\
           &= \max_{S_+ \stackrel{.}{\cup} S_- = [d]} 4 \tr \sqrt{\rho}\Pi_+\sqrt{\rho}\Pi_- ,
\end{align}
where the first maximisation is over distinct basis states $j,k\in [d]$
and $0\leq t \leq 1$; the second over disjoint subsets $S_+$ and $S_-$ 
of $[d]$, with $\Pi_{\bullet} = \sum_{j\in S_{\bullet}} \proj{j}$,
$\bullet = \pm$.

\medskip
For a qubit state $\rho$, we find (see Appendix) that
$C^{(2)}_{\partial\xi}(\rho) = 2 \left|\left(\sqrt{\rho}\right)_{12}\right|^2$ and 
$C^{(\infty)}_{\partial\xi}(\rho) = 4 \left|\left(\sqrt{\rho}\right)_{12}\right|^2$.

\medskip
\textit{F. Largest Shannon information:}
The previous examples should have prepared us for thinking
of visibility as an expression of how much information about 
$\vec{\alpha}$ the output distribution $P(\cdot|\vec{\alpha})$ reveals. 
So why not take this to the logical conclusion? 
Noting that $P$ is a channel from multi-phases
$\vec{\alpha}$ to outputs $\omega$, in the Shannon theoretic sense, 
we are motivated to define visibility as the Shannon capacity of $P$:
\begin{equation}
  \label{eq:shannon-cap}
  V_I[P] := C(P) = \sup_{\mu} I(\vec{\alpha}:\omega),
\end{equation}
where $\mu$ is a probability measure on the $\vec{\alpha}$, defining
a joint distribution $\mu(\vec{\alpha})P(\omega|\vec{\alpha})$
of channel inputs and outputs, and $I(X:Y) = D(\mathbb{P}_{XY}\|\mathbb{P}_X\times\mathbb{P}_Y)$ 
is the mutual information of two random variables \cite{Shannon}.
It can be checked that $V_I$ is regular and weakly affine.
Operationally, $V_I[P]$ is the largest communication rate that
can be transmitted by a sender, who may encode information into the
phase settings $\vec{\alpha}^{(1)},\ldots,\vec{\alpha}^{(n)}$
of asymptotically many interferometers, to a receiver
who decodes the correct message with high probability based on
the observations $\omega_1,\ldots,\omega_n$ \cite{Shannon}.

To obtain $C_I(\rho)$, we then only need to perform a maximisation of
the Shannon capacity over all measurements:
\begin{equation}\begin{split}
  \label{eq:C_I}
  C_I(\rho)  = \sup_{(M_\omega)} C(P_{M|\rho})                     
            &= \sup_{\mu} \sup_{(M_\omega)} I(\vec{\alpha}:\omega) \\
            &= \sup_{\mu} I_{\text{acc}}\bigl(\bigl\{\mu(\vec{\alpha}),\rho(\vec{\alpha})\bigr\}\bigr),
\end{split}\end{equation}
where the latter quantity is know as the \emph{accessible information}.
These optimisations are by no means easy, and are worked out only in
some few cases. In any case, Theorem~\ref{thm:C_V:monotone} shows
that $C_I$ is a SIO monotone. This might provide some motivation to 
try to evaluate $C_I$ in certain special cases.

However, due to the Holevo bound \cite{Holevo:bound}, and 
the Holevo-Schumacher-Westmoreland theorem~\cite{HSW}
regarding the capacity of the cq-channel $\vec{\alpha} \mapsto \rho(\vec{\alpha})$,
we obtain the following:
\begin{equation}
  C_I(\rho) \leq S\bigl(\Delta(\rho)\bigr)-S(\rho) = C_r(\rho)
                                                   = \sup_n \frac1n C_I(\rho^{\ox n}).
\end{equation}

Namely, the Holevo bound \cite{Holevo:bound} upper-bounds the accessible information,
\[\begin{split}
  I_{\text{acc}}\bigl(\bigl\{\mu(\vec{\alpha}),\rho(\vec{\alpha})\bigr\}\bigr)
         &\leq \chi\bigl(\bigl\{\mu(\vec{\alpha}),\rho(\vec{\alpha})\bigr\}\bigr)   \\
         &\!\!\!\!\!\!\!\!\!\!
          :=   S\left(\int \! \mu({\rm d}\vec{\alpha}) \rho(\vec{\alpha})\right)
                      - \int \! \mu({\rm d}\vec{\alpha}) S\bigl( \rho(\vec{\alpha}) \bigr).
\end{split}\]
Here, the second term is always $S(\rho)$ because the $\rho(\vec{\alpha})$
are unitarily rotated versions of $\rho$, and the first term is maximised
by the uniform distribution over all phases:
\begin{equation}
  C_I(\rho) \leq S\bigl(\Delta(\rho)\bigr)-S(\rho) = C_r(\rho),
\end{equation}
with the well-known relative entropy of coherence \cite{Aaberg2006,BCP14}.
Note that the latter is known to be a monotone under IO, and even under the
still larger class of maximally incoherent operations (MIO) \cite{coherence-review}.

Invoking the Holevo-Schumacher-Westmoreland theorem~\cite{HSW}
regarding the capacity of the cq-channel $\vec{\alpha} \mapsto \rho(\vec{\alpha})$,
we get furthermore $\sup_n \frac1n C_I(\rho^{\ox n}) = C_r(\rho)$.

\medskip
In the qubit case, the optimisation (\ref{eq:C_I}) seems to be unknown,
but we believe that the maximum is attained on the binary ensemble 
$\left\{ (\frac12,\rho_0=\rho),\, (\frac12,\rho_1=\sigma_z\rho\sigma_z) \right\}$,
and the measurement in the eigenbasis of $\rho_0-\rho_1$, which would yield
$C_I(\rho) = 1-H\!\left(\!\frac{1\pm 2|\rho_{12}|}{2}\!\right) \approx \frac{2}{\ln 2} |\rho_{12}|^2$.
On the other hand, 
$C_r(\rho) = H\!\left(\!\frac{1\pm\tr\rho\sigma_Z}{2}\!\right) - H\!\left(\!\frac{1\pm r}{2}\!\right)$.

\bigskip
\textbf{Discussion.}---%
Using a simple model of multi-path interferometry and a broad approach
to visibility of an experimental setup of phase modulation and
detection, we showed that the concept of coherence of a state can be 
obtained by optimising the visibility over detection schemes.
We illustrated our approach by analysing specific visibility
functionals. The results are clearest in the two-level case,
corresponding to Mach-Zehnder interferometers, where we find that
the single off-diagonal density matrix element governs almost all 
visibility and coherence effects. In settings with more paths,
as should be expected, there are different inequivalent ways of
quantifying visibility and correspondingly many different, incomparable
coherence measures.

Our discussion shows that it is possible to link coherence theory,
a priori quite an abstract enterprise, to operational notions
in the physics of interferometers. We hope that our present approach
will be fruitful in the future to develop a firm physical foundation
of the resource theory of coherence. As an example of this kind of impact, we 
highlight Theorem~\ref{thm:C_V:monotone}, which shows that visibility
based coherence measures are naturally monotone under strictly
incoherent operations (SIO), while it is an open question whether 
this holds also under the originally proposed incoherent operations (IO);
this might be construed as favouring SIO over IO as the ``correct''
class of operations. See also Yadin \emph{et al.}~\cite{Yadin-etal},
where it is shown that SIO are obtained precisely as the class of
cptp maps that can be dilated in a specific incoherent way onto an
extended system $\mathcal{H}\otimes\mathcal{S}$, where the ``internal'' or 
``spin'' degrees of freedom of the particle are thought of as having no 
incoherence structure, so that in \AA{}berg's framework \cite{Aaberg2006} the
incoherent subspaces are $\ket{j}\otimes\mathcal{S}$. Namely, 
a cptp map is strictly incoherent if and only if it can be
decomposed into attaching an ancillary state of $\mathcal{S}$, followed
by an incoherent unitary on the tensor product space, i.e.~one mapping
the subspaces $\ket{j}\otimes\mathcal{S}$ into each other, followed
by a destructive measurement of $\mathcal{S}$ with outcomes $\lambda$.

Unlike other investigations that have tried to build a similar link
between visibility and coherence, we start from visibility parameters as 
a feature of experimentally accessible data, rather than declaring known 
coherence measures as ``visibility''~\cite{Manab-etal,Emili-etal}. 
Because of this we
think of our approach as operational, in contrast to the cited works
whose approach could be characterised as axiomatic. In this respect,
we believe that our present work goes some way towards answering the
call for an operational justification of coherence as a visibility 
parameter \cite[Sec.~VII]{Coles:coherent}. It is tempting to conjecture
that all the coherence parameters derived from ``reasonable'' visibility 
functionals satisfy duality relations with suitable path-information 
measures such as in the mentioned works. Whether visibility as 
conceptualised by us is always dual to a path-information or some other
parameter, is a question we have to leave open at this point.
In any case, our analysis of some concrete examples of visibility functionals
on interference fringes has bolstered this connection to coherence,
resulting in coherence measures that can be related to, and in
some cases identified with, previously considered measures. 

We also think that the present treatment gives some insight into the relationship
between coherence and the resource theory of asymmetry (or reference
frames) for the group of time translations, 
cf.~\cite{MarvianSpekkens-review}. Namely, looking at Examples
C, D, and E, each of the resulting coherence measures is obtained
by maximising, over a bounded set of diagonal Hamiltonians, a function
given in Eqs.~(\ref{eq:nabla}), (\ref{eq:Fisher}) and (\ref{eq:diff-Chernoff}), 
respectively. It is known that for a fixed Hamiltonian $H$, each of
them is an monotone in the resource theory of time asymmetry corresponding
to energy conservation, see for instance \cite{Girolami:wrong} for the 
latter quantity. 
Hence, we are led to think of coherence theory as asymmetry theory with
a Hamiltonian that has fixed eigenvectors but ``undetermined'' eigenvalues.
This may go some way towards explaining the characteristic similarities
and differences between (time) asymmetry and coherence.

In the analysis we encountered some interesting mathematical problems,
too, among them the characterisation of the set of all commutators
of norm-bounded diagonal and general Hermitian matrices.
Furthermore, we would like to know whether, among the established 
coherence monotones, we can recover the $\ell_1$-measure $C_{\ell_1}$ \cite{BCP14},
or the coherence of formation $C_f$ \cite{Aaberg2006,WinterYang}
directly via visibilities? In light of Refs.~\cite{Manab-etal,Emili-etal},
the former would be especially interesting.

Finally, going beyond the single-particle interference of our above
theory, the present study suggests multi-particle interference as
a natural extension. This will not only provide a framework for the
compositions of systems (cf.~\AA{}berg \cite{Aaberg2006}),
but also bring out the unique quantum features of interference, as opposed 
to the mere wave-mechanical ones in the single-particle case.

\bigskip
\textbf{Ethics statement.}---%
This work did not involve human embryos, nor any other human subjects, nor
data collection on humans.

\bigskip
\textbf{Data accessibility.}---%
This work does not have experimental data.

\bigskip
\textbf{Competing interests.}---%
The authors have no competing interests.

\bigskip
\textbf{Authors' contributions.}---%
All authors have contributed equally and crucially to the conception
of the present paper, the execution of the scientific research, and
its writing. All authors gave final approval for publication.

\bigskip
\textbf{Acknowledgments and funding statement.}---%
It is our pleasure to thank Emili Bagan, Manabendra Bera, John Calsamiglia
and Chang-Shui Yu 
for discussions on interferometers and visibility, and Patrick
Coles for illuminating remarks on an earlier version of the manuscript.

TB is supported by IISER Kolkata and acknowledges the hospitality of
the Quantum Information Group (GIQ) at UAB during June-July 2016, when
the present work was initiated.
MGD is supported by a doctoral studies fellowship of the
Fundaci\'on ``la Caixa''. 
AW is or was supported by the European Commission (STREP ``RAQUEL'')
and the ERC (Advanced Grant ``IRQUAT'').
The authors acknowledge furthermore funding by
the Spanish MINECO (grant FIS2013-40627-P),
with the support of FEDER funds, and by the 
Generalitat de Catalunya CIRIT, project 2014-SGR-966.




\appendix

\phantom{.}\vspace{6.66mm}

\begin{center}
  {\bf APPENDIX --- Qubit examples}
\end{center}

\paragraph{{\bf A.---$\mathbf{C_{\max}}$.}} As only the relative phase $\alpha = \alpha_1-\alpha_2$
matters, we see
\[
  \rho-U(\vec{\alpha})\rho U(\vec{\alpha})^\dagger 
     = \left[\begin{array}{cc}
             0 & (1-e^{-i\alpha})\rho_{12} \\
             (1-e^{+i\alpha})\rho_{21} & 0
             \end{array}\right].
\]
Its trace norm clearly is maximised at $\alpha=\pi$, showing
$C_{\max}(\rho) = |\rho_{12}| + |\rho_{21}| = C_{\ell_1}(\rho)$,
which for qubits is known to equal $2 C_{\tr}(\rho)$.

\medskip
\paragraph{{\bf C.---$\mathbf{C_{\nabla}}$.}} 
For $p=\infty$, the only nontrivial choice is $\Pi_+=\proj{1}$ and 
$\Pi_-=\proj{2}$, directly resulting in 
$C_{\nabla}^{(\infty)} = 2 \bigl\| \proj{1}\rho\proj{2} \bigr\|_1 = 2 |\rho_{12}|$.

For $p=2$, we have to consider the Hamiltonian
$H = \sqrt{t}\proj{1}\pm\sqrt{1-t}\proj{2}$, yielding
\[
  [\rho,H] = \left[\begin{array}{cc}
                   0 & (-\sqrt{t}\pm\sqrt{1-t})\rho_{12} \\
                   (\sqrt{t}\mp\sqrt{1-t})\rho_{21}  & 0
                   \end{array}\right].
\]
Its trace norm is maximised for the negative
sign choice and at $t=\frac12$, and so 
$C_{\nabla}^{(2)} = \sqrt{2} |\rho_{12}|$

\medskip
\paragraph{{\bf D.---$\mathbf{C_{\text{F}}}$.}}
The formula for the coherence measure reduces to
\[
  C_{\text{F}}^{(p)}(\rho) 
    = \max 2\frac{(\lambda_1-\lambda_2)^2}{\lambda_1+\lambda_2} |\bra{e_1}H\ket{e_2}|^2,
\]
where the maximisation is over $H\in\operatorname{span}\{\1,\sigma_Z\}$
such that $\| H \|_p \leq 1$. Note that $\lambda_1+\lambda_2=1$ and
$|\bra{e_1}H\ket{e_2}|^2 = \tr H \proj{e_1} H \proj{e_2}$.

This calculation is conveniently done in the Bloch picture,
writing $\rho = \frac12(\1 + \vec{r}\cdot\vec{\sigma})$, with a vector 
$\vec{r} = r\,\vec{r}^0$ that we decompose as a product of its
length $r=|\vec{r}|$ and a unit vector $\vec{r}^0$ (with components
$r^0_x$, $r^0_y$ and $r^0_z$). In this way the eigenprojectors of $\rho$ become
$\proj{e_{1,2}} = \frac12(\1\pm\vec{r}^0\cdot\vec{\sigma})$.
In the above maximisation, this allows us to identify $\lambda_1-\lambda_2=r$
and $r^2 = 2\tr\rho^2-1$.

For $p=\infty$, we already know that $H=\sigma_Z$ is optimal, so
\[\begin{split}
  C_{\text{F}}^{(\infty)}(\rho) 
           &= 2r^2 \tr \sigma_Z \proj{e_1} \sigma_Z \proj{e_2} \\
           &= 2r^2\frac14\tr (\1-r^0_x\sigma_X-r^0_y\sigma_Z+r^0_z\sigma_Z) \\[-2mm]
           &\phantom{=====:::}
                             \cdot\!(\1-r^0_x\sigma_X-r^0_y\sigma_Z-r^0_z\sigma_Z) \\
           &= r^2\bigl( 1+(r^0_x)^2+(r^0_y)^2-(r^0_z)^2 \bigr) \\
           &= 4\bigl(\tr\rho^2-\tr\Delta(\rho)^2\bigr) \\
           &= 8|\rho_{12}|^2 = 2 C_{\ell_1}(\rho)^2.
\end{split}\]

For $p=2$, the maximisation reduces to that of
$2r^2 \tr H\proj{e_1}H\proj{e_2}$,
with $H=\alpha\1+\beta\sigma_Z$ and $2\alpha^2+2\beta^2\leq 1$. The
trace decomposes into four terms, however the three that contain a $\alpha\1$
evaluate to $0$, leaving $2\beta^2r^2\tr \sigma_Z\proj{e_1}\sigma_Z\proj{e_2}$,
which yields (using the optimal choice $2\beta^2=1$)
$C_{\text{F}}^{(2)}(\rho) = 2(\tr\rho^2-\tr\Delta(\rho)^2) = C_{\ell_1}(\rho)^2$.

\medskip
\paragraph{{\bf E.---$\mathbf{C_{\partial\xi}}$.}}
For $p=\infty$, the only nontrivial choice is $\Pi_+=\proj{1}$ and 
$\Pi_-=\proj{2}$, directly resulting in 
$C_{\partial\xi}^{(\infty)} = 4 \tr \proj{1}\sqrt{\rho}\proj{2}\sqrt{\rho} 
 = 4 \left|\left(\sqrt{\rho}\right)_{12}\right|^2$.

For $p=2$, we have to consider the Hamiltonian 
$H = \sqrt{t}\proj{1}-\sqrt{1-t}\proj{2}$, yielding
\[
  [\sqrt{\rho},H] = \left[\begin{array}{cc}
                            0 & (-\sqrt{t}-\sqrt{1-t})(\sqrt{\rho})_{12} \\
                            (\sqrt{t}+\sqrt{1-t})(\sqrt{\rho})_{21}  & 0
                          \end{array}\right].
\]
Thus,
$I_{\text{WY}}(\rho,H) = \left(\sqrt{t}+\sqrt{1-t}\right)^2 
                           \left|\left(\sqrt{\rho}\right)_{12}\right|^2$,
which is maximised at $t=\frac12$, hence
$C_{\partial\xi}^{(2)} = 2 \left|\left(\sqrt{\rho}\right)_{12}\right|^2$.


\end{document}